\documentclass[11pt]{article}

\usepackage{amsthm,amssymb,amsmath,amsfonts}
\usepackage{enumerate}
\usepackage{verbatim}
\usepackage[top=1.2in,bottom=1.2in,left=1in,right=1in]{geometry}
\usepackage{hyperref}
\usepackage{color}
\usepackage{framed}

\def\final{0}
\ifnum\final=0
\newcommand{\mynote}[1]{\marginpar{\tiny\sf #1}}

\else
\newcommand{\mynote}[1]{}
\fi

\makeatletter
\newtheorem*{rep@theorem}{\rep@title}
\newcommand{\newreptheorem}[2]{%
\newenvironment{rep#1}[1]{%
 \def\rep@title{#2 \ref{##1} (informal)}%
 \begin{rep@theorem}}%
 {\end{rep@theorem}}}
\makeatother

\theoremstyle{plain}
\newtheorem{theorem}{Theorem}%[section]
\newreptheorem{theorem}{Theorem}
\newtheorem{proposition}{Proposition}
\newtheorem{lemma}{Lemma}
\newtheorem{corollary}{Corollary}

\theoremstyle{definition}
\newtheorem{definition}{Definition}

\def \next {\mathsf{next}}

\def \D {\mathcal{D}}

\def \G {\mathcal{G}}
\def \L {\mathcal{L}}
\def \M {\mathcal{M}}

\def \SD {\operatorname{SD}}
\def \cav {\operatorname{cav}}

\def \gets {\leftarrow}

\newcommand{\NN}{\mathbb{N}}

\newcommand{\RR}{\mathbb{R}}
\newcommand{\QQ}{\mathbb{Q}}

\newcommand{\E}{\mathop{\mathbf{E}}}
\newcommand{\eps}{\varepsilon}

\newcommand{\namedref}[2]{#1~\ref{#2}}

\newcommand{\thmlab}[1]{\label{thm:#1}}
\newcommand{\thmref}[1]{\namedref{Theorem}{thm:#1}}

\newcommand{\lemmlab}[1]{\label{lemm:#1}}
\newcommand{\lemmref}[1]{\namedref{Lemma}{lemm:#1}}

\newcommand{\corlab}[1]{\label{cor:#1}}

\newcommand{\seclab}[1]{\label{sec:#1}}
\newcommand{\secref}[1]{\namedref{Section}{sec:#1}}

\newcommand{\applab}[1]{\label{sec:#1}}
\newcommand{\appref}[1]{\namedref{Appendix}{sec:#1}}

\begin{document}

\title{When Can Limited Randomness Be Used in Repeated Games?}
\author{
Pavel Hub\'{a}\v{c}ek\thanks{Weizmann Institute of Science. Supported by the I-CORE Program of the Planning and Budgeting Committee and The Israel Science Foundation (grant No. 4/11). E-mail: \texttt{pavel.hubacek@weizmann.ac.il}.}
\and
Moni Naor\thanks{Weizmann Institute of Science. Incumbent of the Judith Kleeman Professorial Chair. Research supported in part by grants from the Israel Science Foundation, BSF and Israeli Ministry of Science and Technology and from the I-CORE Program of the Planning and Budgeting Committee and the Israel Science Foundation (grant No. 4/11). E-mail: \texttt{moni.naor@weizmann.ac.il}.}
\and
Jonathan Ullman\thanks{Columbia University Department of Computer Science.  Supported by a Junior Fellowship from the Simons Society of Fellows.  Part of this work was done while the author was at Harvard University.  E-mail: \texttt{jullman@cs.columbia.edu}.}
}
%\date{}
\maketitle

\pagenumbering{gobble}
\begin{abstract}
The central result of classical game theory states that every finite normal form game has a Nash equilibrium, provided that players are allowed to use randomized (mixed) strategies.
However, in practice, humans are known to be bad at generating random-like sequences, and true random bits may be unavailable.  Even if the players have access to enough random bits for a single instance of the game their randomness might be insufficient if the game is played many times.

In this work, we ask whether randomness is necessary for equilibria to exist in finitely repeated games.
We show that for a large class of games containing arbitrary two-player zero-sum games, approximate Nash equilibria of the $n$-stage repeated version of the game exist if and only if both players have $\Omega(n)$ random bits.
In contrast, we show that there exists a class of games for which no equilibrium exists in pure strategies, yet the $n$-stage repeated version of the game has an exact Nash equilibrium in which each player uses only a constant number of random bits.

When the players are assumed to be computationally bounded, if cryptographic pseudorandom generators (or, equivalently, one-way functions) exist, then the players can base their strategies on ``random-like'' sequences derived from only a small number of truly random bits.  We show that, in contrast, in repeated two-player zero-sum games, if pseudorandom generators \emph{do not} exist, then $\Omega(n)$ random bits remain necessary for equilibria to exist.

\end{abstract}

\vfill
\pagebreak

%\paragraph{Keywords}
%foo, bar

%--------------------------- Todo list --------------------------------%
\ifnum\final=1
{\color{blue}
  To do:
  \begin{itemize}
  \end{itemize}}
\fi
\vfill
\pagebreak

%--------------------------- TOC -------------------------------------%
\tableofcontents
\vfill
\pagebreak
\pagenumbering{arabic}

%--------------------------- Beginning of Intro --------------------------------%
\section{Introduction}\label{sec:Introduction}
The signature result of classical game theory states that a Nash equilibrium exists in every finite normal form game, provided that players are allowed to play randomized (mixed) strategies. It is easy to see in some games (e.g. Rock-Paper-Scissors) that randomization is necessary for the existence of Nash equilibrium.  However, the assumption that players are able to randomize their strategies in an arbitrary manner is quite strong, as sources of true randomness may be unavailable and humans are known to be bad at generating random-like sequences.

Motivated by these considerations, Budinich and Fortnow~\cite{DBLP:conf/sigecom/BudinichF11} investigated the question of whether Nash equilibria exist when players only have access to \emph{limited randomness}.  Specifically, they looked at the ``repeated matching pennies.''  Matching pennies is a very simple, two-player, two-action, zero-sum game in which the unique equilibrium is for each player to flip a fair coin and play an action uniformly at random.  If the game is repeated for $n$ stages, then the unique Nash equilibrium is for each player to play an independent, uniformly random action in each of the $n$ stages. Budinich and Fortnow considered the case where the players only have access to $\ll n$ bits of randomness, which are insufficient to play the unique equilibrium of the game, and showed that there does not even exist an \emph{approximate} equilibrium (where the approximation depends on the deficiency in randomness).  That is, if the players cannot choose independent, uniformly random actions in each of the $n$ stages, then no approximate equilibrium exists.

In this work, we further investigate the need for randomness in repeated games by asking whether the same results hold for \emph{arbitrary} games.  That is, we start with an arbitrary multi-player game such that Nash equilibria only exist if players can use $\beta$ bits of randomness.  Then we consider the $n$-stage repetition of that game.  Do equilibria exist in the $n$-stage game if players only have access to $\ll \beta n$ bits of randomness?  First, we show that the answer is essentially \emph{no} for arbitrary zero-sum games, significantly generalizing the results of Budinich and Fortnow.  On the other hand, we show that the answer is \emph{yes} for a large class of general games.

These results hold when both players are assumed to be computationally unbounded.  As noted by Budinich and Fortnow, if we assume that the players are required to run in polynomial time, and cryptographic pseudorandom generators (or, equivalently, one-way functions) exist, then a player equipped with only $\ll n$ truly random bits can generate $n$ \emph{pseudorandom bits} that appear truly random to a polynomial time adversary.  Thus, in the computationally bounded regime, if pseudorandom generators exist, then linear randomness is not necessary.  We show that, in contrast, in arbitrary repeated two-player zero-sum games, if pseudorandom generators \emph{do not} exist, then linear randomness remains necessary.

\subsection{Our Results}\label{sec:OurResults}
Suppose we have an arbitrary finite strategic game among $k$ players.  We consider the $n$-stage repetition of this game in which in each of the $n$ consecutive stages, each of the $k$ players simultaneously chooses an action (which may depend on the history of the previous stages).  We assume that in the $1$-stage game $\beta > 0$ bits of randomness for each player are necessary and sufficient for an equilibrium to exist.  We ask whether or not the existence of approximate equilibria in the $n$-stage game requires a linear amount of randomness ($\Omega(n)$ bits) per player.

\paragraph{The case of computationally unbounded players.}
Our first set of results concerns players who are computationally unbounded, which is the standard model in classical game theory.  In this setting, our first result shows that linear randomness is necessary for a large class of games including every two-player zero-sum game.

\begin{reptheorem}{thm:LinearEntropy}
For any $k$-player strategic game in which every Nash equilibrium achieves the minmax payoff profile, in any Nash equilibrium of its repeated version the players' strategies use randomness at least linear in the number of stages.
\end{reptheorem}
\noindent An important subset of strategic games where any Nash equilibrium achieves the minmax payoff profile is the class of two-player zero-sum games where, as implied by the von Neumann's minmax theorem, the concept of Nash equilibrium collapses to the minmax solution.
Hence, to play a Nash equilibrium in any finitely repeated two-player zero-sum game the players must use randomness at least linear in the number of stages.

Second, we show that the above results cannot be extended to arbitrary games.  That is, there exists a class of strategic games that, in their repeated version, admit ``randomness efficient'' Nash equilibria:
\begin{reptheorem}{thm:ConstantEntropy}
For any $k$-player strategic game in which for every player there exists a Nash equilibrium that achieves strictly higher expectation than the minmax strategy, there exists a Nash equilibrium of its repeated version where the players use total randomness independent of the number of stages.
\end{reptheorem}
\noindent
As we shall see, this result is related to the ``finite horizon Nash folk theorem,'' which roughly states that in finitely repeated games every payoff profile in the stage game that dominates the minmax payoff profile can be achieved as a payoff profile of some Nash equilibrium of the repeated game.

\paragraph{The case of computationally efficient players.}
For repeated two-player zero-sum games we study the existence of Nash equilibria with limited randomness when the players are computationally bounded.
Under the assumption that one-way functions do not exist (see the above discussion), we show that it is possible to \emph{efficiently exploit} any opponent (i.e., gain a non-negligible advantage over the value of the stage game) that uses low randomness in every repeated two-player zero-sum game.
Hence, in repeated two-player zero-sum games there are no computational Nash equilibria in which one of the players uses randomness sub-linear in the number of the stages.
\begin{reptheorem}{thm:EfficientAdvantage}
In any repeated two-player zero-sum game, if one-way functions do not exist, then for any strategy of the column player using sub-linear randomness, there is a computationally efficient strategy for the row player that achieves an average payoff non-negligibly higher than his minmax payoff in the stage game.
\end{reptheorem}
\noindent
The proof of this result employs the algorithm of Naor and Rothblum~\cite{DBLP:conf/icml/NaorR06} for learning adaptively changing distributions.
The main idea is to adaptively reconstruct the small randomness used by the opponent in order to render his strategy effectively deterministic and then improve the expectation by playing the best response.

\paragraph{Strong exploitation of low-randomness players.} In the classical setting, i.e., without restrictions on the computational power of the players, it was shown by Neyman and Okada~\cite{DBLP:journals/geb/NeymanO00} that in every repeated two-player zero-sum game it is possible to extract utility proportional to the randomness deficiency of the opponent.
On the other hand, our result in the setting with computationally efficient players guarantees only a non-negligible advantage in the presence of a low-randomness opponent.
This leaves open an intriguing question of how much utility can one \emph{efficiently} extract from an opponent that uses low randomness in a repeated two-player zero-sum game (see \secref{StrongExploitation} for additional discussion).

\paragraph{The case of matching pennies.}
As noticed by Budinich and Fortnow~\cite{DBLP:conf/sigecom/BudinichF11}, the repeated game of matching pennies exhibits clear tradeoffs between the randomness available to players and existence of $\eps$-Nash equilibria.
Our work generalizes their results already in the context of repeated matching pennies, since they assumed that the players randomize their strategies by flipping limited number of coins, whereas we only assume that the players' strategies are of low entropy.
Our results for the game of matching pennies are provided in \appref{MatchingPennies}.

\subsection{Other Related Work}\label{sec:RelatedWork}
In one of the first works to consider the relation between the randomness available to players and the existence of equilibria Halpern and Pass~\cite{halpern2014algorithmic} introduced a computational framework of machine games that explicitly incorporates the cost of computation into the utility functions of the players and specifically the possibility of randomness being expensive.
They demonstrated this approach on the game of Rock-Paper-Scissors, and showed that in machine games where randomization is costly then Nash equilibria do not necessarily exist. However, in machine games where randomization is free then Nash equilibria always exist.

Based on derandomization techniques, Kalyanaraman and Umans~\cite{kalyanaraman2007algorithms} proposed randomness efficient algorithms both for finding equilibria and for playing strategic games.
In the context of finitely repeated two-player zero-sum games where one of the players (referred to as the learner) is uninformed of the payoff matrix, they gave an adaptive on-line algorithm for the learner that can reuse randomness over the stages of the repeated game.

Halprin and Naor~\cite{DBLP:journals/crossroads/HalprinN10} suggested the possibility of using randomness generated by human players in repeated games for generation of pseudorandom sequences.
The strategic game they proposed for this purpose is a zero-sum two-player game.
As shown by our results, their choice improves the likelihood of extracting truly random bits from the gameplay, since the players must use linear randomness in the number of stages in equilibria of any repeated two-player zero-sum game.

%--------------------------- End of Intro -------------------------------------%

%--------------------------- Beginning of Notation -----------------------------%
\section{Notation and Background}\label{sec:Notation}
\subsection{Game Theoretic Background}
Here we provide the concepts from game theory that we use in this work (for an in-depth study see the classical text by Osborne and Rubinstein~\cite{osborne1994course}).

\begin{definition}[strategic game]\label{def:strategicGame}
A \emph{strategic game} $G=\langle N, (A_i), (u_i)\rangle$ is a tuple consisting of
%a finite set of players $N$,
%for each player $i\in N$ a nonempty set of actions $A_i$, and
%for each player $i\in N$ a utility function $u_i:A\rightarrow\RR$ assigning each action profile $a\in A=\times_{j\in N}A_j$ a real-valued payoff $u_i(a)$.
\begin{itemize}
\item a finite set of players $N$
\item for each player $i\in N$ a nonempty set of actions $A_i$
\item for each player $i\in N$ a utility function $u_i:A\rightarrow\RR$ assigning each action profile $a\in A=\times_{j\in N}A_j$ a real-valued payoff $u_i(a)$.
\end{itemize}
\end{definition}
In the special case when $G$ is a two-player zero-sum game we use the notation $\langle (A_1,A_2),u\rangle$ instead of $\langle \{1,2\},(A_1,A_2),(u_1,u_2)\rangle$, since there are only two players and $u_1(a)=-u_2(a)$ for all $a\in A_1\times A_2$.
We refer to player 1 as the row player (also known as Rowena) and to player 2 as the column player (also known as Colin).\footnote{We have adopted Colin and Rowena from Aumann and Hart~\cite{aumann2003long}.}

We denote by $S_i$ the set of mixed strategies of player $i$, i.e., the set $\Delta(A_i)$ of all probability distributions on the action space of player $i$. For a strategy profile $\sigma\in S=\times_{j\in N}S_j$ we use $\sigma_i$ to denote the strategy of player $i$ in $\sigma$ and $\sigma_{-i}$ to denote the profile of strategies of all the players in $N$ except for player $i$ in $\sigma$, and we write $\sigma$ equivalently as $(\sigma_i,\sigma_{-i})$.

\begin{definition}[Nash equilibrium in strategic game]
A \emph{Nash equilibrium} of a strategic game $\langle N, (A_i), (u_i)\rangle$ is a profile $\sigma$ %$\sigma \in S = \times_{j\in N}S_j$
of strategies with the property that for every player $i\in N$ we have
%$\E[u(\sigma_i,\sigma_{-i})] \geq \E[(\sigma'_{i}, \sigma_{-i})]$ for all $\sigma'_i \in S_i$.
 \[
 \E[u(\sigma_i,\sigma_{-i})] \geq \E[(\sigma'_{i}, \sigma_{-i})] \text{ for all } \sigma'_i \in S_i\ .
 \]
\end{definition}

\begin{definition}[minmax payoff]
The \emph{minmax payoff of player $i$} in strategic game $\langle N, (A_i), (u_i)\rangle$, denoted $v_i$, is
the lowest payoff that the other players \emph{can force upon player $i$}, i.e.,
%$v_i = \min_{\sigma_{-i}\in S_{-i}}\max_{\sigma_i\in S_i} \E[u_i (\sigma_{i},\sigma_{-i})]$.
\[
v_i = \min_{\sigma_{-i}\in S_{-i}}\max_{\sigma_i\in S_i} \E[u_i (\sigma_{i},\sigma_{-i})]\ .
\]
A \emph{minmax strategy of player $i$} in $G$ is a strategy $\hat{\sigma}_i\in S_i$ such that $ \E[u_i (\hat{\sigma}_{i},\sigma_{-i})]\geq v_i$ for all $\sigma_{-i}\in S_{-i}$.
\end{definition}

\begin{definition}[feasible and individually rational payoff profile]
An \emph{individually rational payoff profile} of $G$ is a vector $p\in\RR^{|N|}$ that weakly dominates the minmax payoff of every player, i.e., a vector for which $p_i\geq v_i$ for all $i\in N$.
A vector $p\in\RR^{|N|}$ is a \emph{feasible payoff profile of $G$} if there exists a collection $\{\alpha_a\}_{a\in A}$ of nonnegative rational numbers such that
$\sum_{a\in A}\alpha_a = 1$ and $p_i = \sum_{a\in A}\alpha_a u_i(a)$ for all $i\in N$.
\end{definition}
Note that since in every finite strategic game a Nash equilibrium always exists, there also always exists an individually rational payoff profile (the payoff profile of the Nash equilibrium). However, the Nash equilibrium payoff profile is not necessarily feasible in the above sense.

\begin{definition}[$n$-stage repeated game]
Let $G=\langle N, (A_i), (u_i)\rangle$ be a strategic game. An \emph{$n$-stage repeated game of $G$} is an extensive form game with perfect information and simultaneous moves $G^n=\langle N,H,P,(u^{*}_i)\rangle$ in which:
\begin{itemize}
\item $H=\{\emptyset\}\cup\{\bigcup_{t=1}^{n}A^{t}\}$, where $\emptyset$ is the initial history and $A^{t}$ is the set of sequences of action profiles in $G$ of length $t$
\item $P(h)=N$ for each non-terminal history $h\in H$
\item $u^{*}_i(a^1,\ldots ,a^n)=\frac{1}{n}\sum_{t=1}^{n}u_i(a^t)$ for every terminal history $(a^1,\ldots ,a^n)\in A^n$.
\end{itemize}
A \emph{behavioral strategy of player $i$} is a collection $(\sigma_i(h))_{h\in H\setminus A^n}$ of independent probability measures (one for each non-terminal history), where each $\sigma_i(h)$ is a probability measure over $A_i$.
\end{definition}

\begin{definition}[Nash equilibrium in $n$-stage repeated game]
A \emph{Nash equilibrium} of an $n$-stage repeated game of $G=\langle N, (A_i), (u_i)\rangle$ is a profile $\sigma$ of behavioral strategies with the property that for every player $i\in N$ and every behavioral strategy $\sigma'_i$, we have
%$\E[u^{*}(\sigma_i,\sigma_{-i})] \geq \E[u^{*}(\sigma'_{i}, \sigma_{-i})]$.
 \[
 \E[u^{*}(\sigma_i,\sigma_{-i})] \geq \E[u^{*}(\sigma'_{i}, \sigma_{-i})]\ .
 \]
\end{definition}

\subsection{Cryptographic Background}\label{sec:CryptoNotation}

\paragraph{Pseudorandom generators and one-way functions.}
The notion of cryptographic pseudorandom generators was introduced by Blum and Micali~\cite{DBLP:journals/siamcomp/BlumM84}, who defined them as algorithms that produce sequences of bits unpredictable in polynomial time, i.e., no efficient next-bit-test is able to predict the next output of the pseudorandom generator given the sequence of bits generated so far.
As Yao~\cite{DBLP:conf/focs/Yao82a} showed, this is equivalent to  a generator whose output is indistinguishable from a truly random string to any polynomial time observer.
One of the central questions in cryptography is to understand the assumptions that are sufficient and necessary for implementing a particular cryptographic task.
Impagliazzo and Luby~\cite{DBLP:conf/focs/ImpagliazzoL89} (see also Impagliazzo~\cite{impagliazzo1992thesis}) showed that one-way functions are essential for many cryptographic primitives (e.g., private-key encryption, secure authentication, coin-flipping over telephone).
H\aa{}stad, Impagliazzo, Levin and Luby~\cite{DBLP:journals/siamcomp/HastadILL99} showed that pseudorandom generators exist if and only if one-way functions exist.
Therefore the existence of one-way functions is the major open problem of cryptography.
For an in depth discussion see Goldreich~\cite{DBLP:books/cu/Goldreich2001}.

\paragraph{Standard notation.}A function $\mu:\NN\rightarrow\RR^+$ is \emph{negligible} if for all $c\in\NN$ there exists $n_c\in\NN$ such that for all $n \geq n_c$, $\mu(n) \leq n^{-c}$.
A function $\mu:\NN\rightarrow\RR^+$ is \emph{noticeable} if there exists $c\in\NN$ and $n_c\in\NN$ such that for all $n \geq n_c$, $\mu(n) \geq n^{-c}$.

\begin{definition}[statistical distance]
The statistical distance between two distributions $X$ and $Y$ over $\{0,1\}^\ell$, denoted by $\SD(X,Y)$, is defined as:
\[
\SD(X,Y) = \frac{1}{2}\sum_{\alpha\in\{0,1\}^\ell}{\left|\Pr[X = \alpha] - \Pr[Y = \alpha]\right|}\ .
\]
\end{definition}

The most fundamental notion for measuring randomness is the Shannon entropy:

\begin{definition}[Shannon entropy]
Given a probability distribution $\rho \in \Delta(A)$, the \emph{Shannon entropy of $\rho$} is defined as
$$
H(\rho) := \E_{a \gets \rho}\left( \log_2 \left(\frac{1}{\Pr\left(\rho = a\right)} \right)\right).
$$
\end{definition}

As mentioned above, if we have a one-way function then many cryptographic primitives are possible and in particular we can stretch a short seed into a long seemingly random one. Hence, we will be interested in the case that such functions do not exist.

\begin{definition}[almost one-way function]
A function $f$ is an \emph{almost one-way function} if it is computable in polynomial time,
and for infinitely many input lengths, for any PPTM $\M$, the probability that $\M$ inverts $f$ on a random input is negligible.
Namely, for any polynomial $p$, there exist infinitely many choices of $n \in \mathbb{N}$ such that
\[
\Pr_{x\sim U_{k(n)},\,\M}[\M(f(x)) \in f^{-1}(x)] < \frac{1}{p(n)}\ .
\]
\end{definition}

%--------------------------- End of Notation ----------------------------------%

%--------------- Beginning of Nash Equilibria with Low Entropy ----------------%

%-------------------------- Low Entropy Nash ----------------------------------%
\section{Low-Entropy Nash Equilibria of Finitely Repeated Games}\label{sec:unbounded}
In this section we show that, in the setting with players that have unbounded computational power, there are two classes of $k$-player strategic games at the opposite sides of the spectrum with respect to the amount of randomness necessary for equilibria of their repeated versions.

To measure the randomness of a player's strategy we consider the maximal total Shannon entropy of his strategies used along any terminal history.

\begin{definition}[Shannon entropy of a strategy in repeated game]\label{def:EntropyOfStrategy}
Let $G=\langle N, (A_i), (u_i)\rangle$ be a finite strategic game and let $\sigma_i$ be a strategy of player $i$ in the $n$-stage repeated game of $G$. For any terminal history $a=(a^1,\ldots,a^{n})\in A^{n}$, let $(\sigma_i(\emptyset),\sigma_i(a^1),\sigma_i(a^1,$ $a^2),\ldots,\sigma_i(a^1,\ldots,a^{n-1}))$ be the $n$-tuple of strategies of player $i$ in $\sigma_i$ at all the non-terminal subhistories of $a$. We define the \emph{Shannon entropy of $\sigma_i$}, denoted as $H(\sigma_i)$, as
\[
H(\sigma_i):=\max_{a\in A^{n}}\left\{H(\sigma_i(\emptyset))+\sum_{j=1}^{n-1}{H(\sigma_i(a^1,\ldots,a^j))}\right\}\ .
\]
\end{definition}
This is a worst case notion, in that it measures the entropy of the strategy of player $i$ irrespective of the strategies of the other players.
For some of our results we consider its alternative variant of \emph{effective Shannon entropy of a strategy $\sigma_i$ in a strategy profile $\sigma$}, i.e.,
the maximal total entropy of $\sigma_i$ along terminal histories that are sampled in $\sigma$ with non-zero probability.

For the restricted class of games in which any Nash equilibrium payoff profile is exactly the minmax payoff profile (e.g. any two-player zero-sum game), the following proposition relates the Nash equilibria of the strategic game to the structure of Nash equilibria in its $n$-stage repeated version.\footnote{A variant of Proposition~\ref{prop:OptimalEntropy} with respect to pure equilibria is given in Osborne and Rubinstein~\cite{osborne1994course} as Proposition 155.1.}

\begin{proposition}\label{prop:OptimalEntropy}
Let $G=\langle N, (A_i), (u_i)\rangle$ be a strategic game such that any Nash equilibrium payoff profile is equal to the minmax payoff profile.
For all $n\in\NN$, if $\sigma$ is a Nash equilibrium of $G^n=\langle N,H,P,(u^{*}_i)\rangle$, the $n$-stage repeated game of $G$, then for every non-terminal history $h\in H$ sampled with non-zero probability by $\sigma$ the strategy profile $\sigma(h)$ is a Nash equilibrium of $G$.
\end{proposition}
\begin{proof}
Assume to the contrary that there exists a Nash equilibrium $\sigma$ of $G^n$ such that for some non-terminal history $h\in H$, sampled with non-zero probability by $\sigma$, the strategy profile $\sigma(h)$ is not a Nash equilibrium of $G$.
 Let $h$ be without loss of generality the longest history such that $\sigma(h)$ is not a Nash equilibrium of $G$.
 There exists a player $i$ with a profitable deviation $\sigma^{*}_{i}$ in the stage game to his strategy in the strategy profile $\sigma(h)$.
 Consider the strategy $\sigma'_i$ of player $i$ in $G^n$ defined in the following way: $\sigma'_i(h')=\sigma_i(h')$ for any history $h'\in H$ that does not contain $h$ as a subhistory, $\sigma'_i(h)=\sigma^{*}_{i}$ for the history $h$, and $\sigma'_i(h'')$ is the minmax strategy $\hat{\sigma}_i$ of player $i$ in $G$ for any history $h''\neq h$ that contains $h$ as a subhistory.

 Note that for any history $h'\in H$ that does not contain $h$ as a subhistory, $\E[u_i((\sigma'_i,\sigma_{-i})(h'))]=\E[u_i((\sigma_i,\sigma_{-i})(h'))]$ by the construction of $\sigma'_i$.
 Since the minmax strategy $\hat{\sigma}_i$ of player $i$ guarantees at least the minmax payoff $v_i$ (equal to any Nash equilibrium payoff of player $i$ in $G$),
 $\E[u_i((\sigma'_i,\sigma_{-i})(h''))]\geq\E[u_i((\sigma_i,\sigma_{-i})(h''))]$ for any history $h''\neq h$ that contains $h$ as a subhistory.
 Finally, $\E[u_i((\sigma'_i,\sigma_{-i})(h))]>\E[u_i((\sigma_i,\sigma_{-i})(h))]$ because $\sigma^{*}_{i}$ is a profitable deviation for player $i$ in $G$ given the strategy profile $\sigma(h)$.

Recall that the history $h$ is sampled in $\sigma$ with non-zero probability, and hence $\E[u^{*}_{i}(\sigma'_i,\sigma_{-i})]>\E[u^{*}_{i}(\sigma_i,\sigma_{-i})]$, i.e.,
 the alternative strategy $\sigma'_i$ increases the expectation of player $i$ in $G^n$ given that the other players follow $\sigma_{-i}$, a contradiction to $\sigma$ being a Nash equilibrium of $G^n$.
\end{proof}

For strategic games from this class, Proposition~\ref{prop:OptimalEntropy} immediately gives a linear lower bound on entropy needed to play Nash equilibria in their repeated games.

\begin{theorem}\thmlab{LinearEntropy}
Let $G$ be a strategic game such that any Nash equilibrium payoff profile is equal to the minmax payoff profile. For all $n\in\NN$ and every player $i\in N$, if in any Nash equilibrium of $G$ the strategy of player $i$ is of entropy at least $\beta_i$ then in any Nash equilibrium of the $n$-stage repeated game of $G$ the strategy of player $i$ is of entropy at least $n\beta_i$.
\end{theorem}
\begin{proof}
Assume to the contrary that there exists a Nash equilibrium $\sigma$ of the $n$-stage repeated game of $G$ with strategy of entropy strictly smaller than $n\cdot\beta_i$ for player $i$.
By Proposition~\ref{prop:OptimalEntropy}, $\sigma(h)$ is a Nash equilibrium of $G$ for all $h$ sampled by $\sigma$ with non-zero probability. Hence, there must exist a history $h^*\in H$ sampled with non-zero probability in $\sigma$ such that $\sigma(h^*)$ is a Nash equilibrium of $G$ and the entropy $H(\sigma_{i}(h^{*}))$ of $\sigma_{i}(h^{*})$ is strictly smaller than $\beta_i$, a contradiction.
\end{proof}

\begin{figure}[h]%
\def\arraystretch{1.5}
\begin{center}
\begin{tabular}{r|c|c|c|c|}
\multicolumn{1}{r}{}
 & \multicolumn{1}{c}{Left ($L$)}
 & \multicolumn{1}{c}{Heads ($H$)}
 & \multicolumn{1}{c}{Tails ($T$)}
 & \multicolumn{1}{c}{Right ($R$)}\\
\cline{2-5}
~~~Up ($U$)~~ & $~~0,-1$ & $~~0,-1$ & $~~0,-1$ & $~~0,~~0$ \\
\cline{2-5}
~~~Heads ($H$)~~ & $~~0,-1$ & $~~1,-1$ & $-1,~~1$ & $-1,~~0$ \\
\cline{2-5}
~~~Tails ($T$)~~ & $~~0,-1$ & $-1,~~1$ & $~~1,-1$ & $-1,~~0$\\
\cline{2-5}
~~~Down ($D$)~~ & $~~0,~~0$ & $-1,~~1$ & $-1,~~1$ & $~~1,~~0$ \\
\cline{2-5}
\end{tabular}
\end{center}
\caption{The payoff matrix of an extended game of matching pennies.}%
\label{fig:ExtendedMatchingPennies}%
\end{figure}

\paragraph{Repeated non-zero-sum game requiring a lot of randomness.}
\thmref{LinearEntropy} applies not only to two-player zero-sum games but also to some non-zero-sum games.
The game $G$ given by the payoff matrix in Figure~\ref{fig:ExtendedMatchingPennies} is a variant of the game of matching pennies where the players have two additional options.
There are three mixed Nash equilibria in $G$: $(\frac{1}{2}H+\frac{1}{2}T,\frac{1}{2}H+\frac{1}{2}T)$, $(\frac{1}{2}U+\frac{1}{2}D,\frac{1}{2}H+\frac{1}{2}R)$, and $(\frac{1}{2}U+\frac{1}{2}D,\frac{1}{2}T+\frac{1}{2}R)$; all the three Nash equilibria achieve the same payoff profile $(0,0)$ and require each player to use one random bit.
Notice that the row player can get utility $0$ irrespective of the strategy of the column player by selecting his action ``Up'', and similarly the column player can ensure utility $0$ by playing ``Right''.
Hence, the minmax payoff profile is $(0,0)$.
Since none of the three Nash equilibria of $G$ improves over the minmax payoff profile, we get by \thmref{LinearEntropy} that each player must use strategy of entropy at least $n$ in any Nash equilibrium of the $n$-stage repeated game of $G$.

\paragraph{Repeated non-zero-sum game requiring low randomness.}
On the other hand, there are strategic games for which \thmref{LinearEntropy} does not apply, and the players may use in the $n$-stage repeated game equilibrium strategies of entropy proportional only to the entropy needed in the single-shot game.

\begin{figure}[t]%
\def\arraystretch{1.5}
\begin{center}
\begin{tabular}{r|c|c|c|c|}
\multicolumn{1}{r}{}
 &  \multicolumn{1}{c}{Cooperate ($C$)}
 & \multicolumn{1}{c}{Heads ($H$)}
 & \multicolumn{1}{c}{Tails ($T$)}
 & \multicolumn{1}{c}{Punish ($P$)}\\
\cline{2-5}
~~~Cooperate ($C$)~~~ & $~~3,~~3$ & $-3,~~6$ & $-3,~~6$ & $-3,-3$ \\
\cline{2-5}
~~~Heads ($H$)~~~ & $~~6,-3$ & $ ~~1,-1$ & $-1,~~1$ & $-3,-3$ \\
\cline{2-5}
~~~Tails ($T$)~~~ & $~~6,-3$ & $-1,~~1$ & $ ~~1,-1$ & $-3,-3$ \\
\cline{2-5}
~~~Punish ($P$)~~~ & $-3,-3$ & $-3,-3$ & $-3,-3$ & $-4,-4$ \\
\cline{2-5}
\end{tabular}
\end{center}
\caption{The payoff matrix of an extended game of matching pennies.}%
\label{fig:PenniesWithPunishment}%
\end{figure}

Consider for example the strategic game $G$ given by the payoff matrix in Figure~\ref{fig:PenniesWithPunishment}.
The strategy profile $\sigma=(\frac{1}{2}H+\frac{1}{2}T,\frac{1}{2}H+\frac{1}{2}T)$ is the unique Nash equilibrium of $G$ that achieves payoff profile $(0,0)$.
The minmax payoff profile is $(-3,-3)$, since any player can get utility at least $-3$ by playing $C$.
We show that the $n$-stage repeated game of $G$ admits a Nash equilibrium that requires only a single random coin, i.e.,
the same amount of randomness as the Nash equilibrium $\sigma$ of the stage game $G$.
Consider the strategy profile in which both players play $C$ in the first $n-1$ rounds and in the last round each player plays $H$ and $T$ with equal probability,
and if any player deviates from playing $C$ in one of the first $n-1$ rounds then the opponent plays $P$ throughout all the remaining stages.
To see that this strategy profile is a Nash equilibrium of the $n$-stage repeated game of $G$ note that any deviation from playing $C$ in the first $n-1$ rounds can increase the utility of any player by at most $3$ (by playing either $H$ or $T$ instead of $C$),
however the subsequent punishment induces a loss of at least $-3$ which renders any deviation unprofitable.

The randomness efficient Nash equilibrium from the above example resembles the structure of Nash equilibria constructed in the proof of the \emph{Nash folk theorem for finitely repeated games}.
This theorem characterizes the payoff profiles that can be achieved by Nash equilibria of the repeated game.
In particular, it shows that in strategic games $G$ such that for very player $i$ there exists a Nash equilibrium $\sigma_i$ strictly improving over his minmax payoff any feasible payoff profile (i.e., any convex combination of payoff profiles in $G$ with rational coefficients) that is individually rational (i.e., achieves at least the minmax level for every player) can be approximated by a Nash equilibrium of sufficiently long finitely repeated game of $G$ (cf.\ Osborne and Rubinstein~\cite{osborne1994course} for a survey of known folk theorems).

The main idea behind the proof of the folk theorem is that for every player $i$ the gap between the payoff in the Nash equilibrium $\sigma_i$ and the minmax payoff $v_i$ can be used to punish the player in case he deviates from the strategy that approximates any feasible and individually rational payoff profile. In particular, in any such Nash equilibrium the players use a fixed number of rounds (independent of the number of stages $n$) before the last round in which they play according to some (possibly mixed) Nash equilibria of the stage game and in the preceding rounds they play pure strategies so that the overall payoff approximates the feasible payoff profile.
Hence, the amount of randomness on all the equilibrium paths is independent of the number of stages in any such Nash equilibrium of the repeated game.

\begin{theorem}\thmlab{ConstantEntropy}
Let $G$ be a strategic game such that for every player $i$ there exists a Nash equilibrium $\sigma_i$ of $G$ in which the payoff of player $i$ exceeds his minmax payoff $v_i$ and there exists a feasible and individually rational payoff profile in $G$. Let $\beta_i$ be such that in any Nash equilibrium of $G$ the strategy of player $i$ is of entropy at most $\beta_i$.
There exists $c\in\NN$ such that for all sufficiently large $n\in\NN$ and every player $i\in N$ there exists a Nash equilibrium of $G^n$, the $n$-stage repeated game of $G$, in which the strategy of player $i$ is of \emph{effective entropy} at most $c\cdot\beta_i$.
\end{theorem}
\begin{proof}
Let $p\in\RR^{|N|}$ be the feasible and individually rational payoff profile of $G$.
There exist coefficients $\{\alpha_a\}_{a\in A}\subset\QQ$ such that $\sum_{a\in A}\alpha_a=1$ and for all $i\in N$, $p_i=\sum_{a\in A}\alpha_a u_i(a)$.
Let $K$ be the smallest integer such that each $\alpha_a$ can be written as $\alpha'_a/K$ for $\alpha'_a\in\NN$.
For some $\ell\in\NN$, we divide the stages in $G^n$ into two parts of length $\ell\cdot K$ and $m=n-\ell\cdot K$.
Let $s$ be a strategy profile in $G^n$ that schedules the first $\ell\cdot K$ stages such that each action profile $a$ for which $\alpha_a\neq 0$ is played by the players in exactly $\ell\cdot\alpha'_a$ number of stages.
In the remaining $m$ stages the players cycle between the Nash equilibria $\{\sigma_i\}_{i\in N}$, i.e.,
for all $j\in\{0,\ldots,m-1\}$ at the stage $n-m+1+j$ the players play the Nash equilibrium $\sigma_{j'}$, where $j'=1+(j\mod|N|)$.
In case any player $i$ deviates from $s$ in one of the first $\ell\cdot K$ rounds, the remaining players play the strategy that forces the minmax level $v_i$ on player $i$.

Note that if the number $m$ of the last stages is such that for all action profiles $a\in A$ with $\alpha_a\neq 0$ and for every player $i$:
\[
\frac{m}{|N|}\left(\sum_{j\in N}{\E[u_i(\sigma_j)]}-|N|v_i \right)
\geq
\max_{a'_i\in A_i} u_i(a'_i,a_{-i})-u_i(a)\ ,
\]
then no player has a profitable deviation and $\sigma$ is a Nash equilibrium of $G^n$.
The number $m$ of last stages can be bounded by some constant $c$ selected independently of $n$.
Since the number of stages in which the players play according to some Nash equilibrium of $G$ is at most $c$ (the players take pure actions in all the first $n-c$ stages), for any player $i$ the effective entropy of $s_i$ in $s$ is at most $c\cdot\beta_i$.
\end{proof}

\paragraph{Randomness in Subgame Perfect Equilibria of Finitely Repeated Games.}
An unavoidable shortcoming of the solution concept of Nash equilibrium in the context of repeated (and in general extensive form) games is that it is possible for equilibria to be established based on non-credible threats.
This issue can be circumvented by the stronger requirement of \emph{subgame perfection} that demands the players' strategies to be best response at every history (even off the equilibrium path), and hence implicitly eliminates all empty threats.

Since any subgame perfect equilibrium is a Nash equilibrium, the linear lower bound on the amount of entropy applies to subgame perfect equilibria when the minmax payoff profile cannot be improved upon by any Nash equilibrium in the stage game.
On the other hand, it is possible to construct a randomness efficient subgame perfect equilibrium in the $n$-stage repeated game if in the underlying game there are two Nash equilibria with different payoffs for each player. Such subgame perfect equilibrium is constructed in the proof of perfect finite horizon Folk theorem of Beno\^{i}t and Krishna~\cite{benoit1985finitely}.

\paragraph{Characterization of games with randomness efficient equilibria.}
The condition on the structure of the stage game in \thmref{ConstantEntropy} (i.e., that for every player there exists a Nash equilibrium of the stage game that strictly improves over his minmax payoff) is the same as in the Nash Folk theorem of Beno\^{i}t and Krishna~\cite{benoit1987nash}.
We leave it as an open problem whether ideas from a proof of a more general finite horizon Nash folk theorem (e.g. the one given by Gonz{\'{a}}lez{-}D{\'{\i}}az~\cite{DBLP:journals/geb/Gonzalez-Diaz06}) could help extend (or characterize) the class of games that admit randomness efficient equilibria in their repeated versions.
%------------------------ End of Low Entropy Nash -----------------------%

%--------------------- Low Entropy Computational Nash -------------------%
\section{Low-Entropy Computational Nash Equilibria of Finitely Repeated Two-Player Zero-Sum Games}\label{sec:ComputationalGeneral}
In this section we study randomness in equilibria of repeated two-player zero-sum games with computationally efficient players.
The solution concept we consider in this setting is \emph{computational Nash equilibrium} (introduced in the work of Dodis, Halevi and Rabin~\cite{DBLP:conf/crypto/DodisHR00}) that assumes that the players are restricted to computationally efficient strategies and indifferent to negligible improvements in their utilities, i.e., a computational Nash equilibrium is analogous to the concept of $\eps$-Nash equilibrium with a negligible $\eps$, where the player's strategies, as well as any deviations, must be computationally efficient.

To capture the requirement of computational efficiency, the players' strategies must be implemented by families of polynomial-size circuits.
For a two-player zero-sum game $G$, we denote by \emph{repeated game of $G$} the infinite collection $\{G^n\}_{n\in\NN}$ of all the $n$-stage repeated games of $G$. A family of polynomial size circuits $\{C_n\}_{n\in\NN}$ implements the strategy of the row player in the repeated game of $G$ as follows. In $G^n$, the $n$-stage repeated game of $G$, the circuit $C_n$ takes as input a string corresponding to a non-terminal history $h$ in $G^n$ and $s(n)$ random bits; it outputs an action to be taken at history $h$.
If the strategy of player $i\in\{1,2\}$ is implemented by family $\{C_{n}^{i}\}_{n\in\NN}$ then the gameplay in the $n$-stage repeated game of $G$ is defined in the following way: player $i$ samples a random string $r_i\in\{0,1\}^{s_i(n)}$ and at each stage of $G^n$ takes the action $a=C_{n}^{i}(h,r_i)\in A_i$, given that the history of play up to the current stage is $h$. The utility function $u_{n}^{*}$ is for all $n$ defined as in the standard $n$-stage repeated game of $G$ (i.e., it is the average utility achieved in the stage game over the $n$ stages).

\begin{definition}[computational Nash equilibrium of repeated game]
For a two-player zero-sum game $G=\langle (A_1,A_2),u\rangle$, \emph{a computational Nash equilibrium of the repeated game of $G$} is a strategy profile $(\{C_{n}^{1}\}_{n\in\NN},\{C_{n}^{2}\}_{n\in\NN})$ given by polynomial-size circuit families such that for every player $i\in\{1,2\}$ and every strategy $\{\tilde{C}_{n}^{i}\}_{n\in\NN}$ given by a polynomial-size circuit family it holds for all large enough $n\in\NN$ that
 \[
 \E[u^{*}_{n}(C_{n}^{i},C_{n}^{-i})] \geq \E[u^{*}_{n}(\tilde{C}_{n}^{i},C_{n}^{-i})]+\eps(n)\ ,
 \]
where $\eps$ is a negligible function.
\end{definition}

We show that if one-way functions do not exist, then in repeated two-player zero-sum games there are no computational Nash equilibria in which the players' strategies use random strings of length sub-linear in the number of the stages.

Our result follows by showing that
finding efficiently a best response to the opponent's strategy that uses limited randomness can be seen as a special case of the problem of \emph{learning an adaptively changing distribution} (introduced by Naor and Rothblum~\cite{DBLP:conf/icml/NaorR06}).
The goal in their framework is for a learner to recover a secret state used to sample a publicly observable distribution, in order to be able to predict the next sample.
In particular, this would allow the learner to be competitive to someone who knows the secret state (Naor and Rothblum~\cite{DBLP:conf/icml/NaorR06} considered this problem in the context of an adversary trying to impersonate someone in an authentication protocol).
In the setting of repeated games, the random string used by the opponent's strategy can be thought of as the secret state. Note that learning it at any non-terminal history would give rise to efficient profitable deviation, since the player could just compute the next move of his opponent and play the best response to it.

\paragraph{Learning adaptively changing distributions.}
An adaptively changing distribution is given by a pair of algorithms $\G$ and $\D$ for generating an initial state and sampling. The algorithm $\G$ is a randomized function $\G:R\rightarrow S_p \times S_{init}$ that outputs an initial public state $p_0$ and a secret state $s_0$. The sampling algorithm $\D$ is a randomized function $\D: S_p \times S_s \times R \rightarrow S_p \times S_s$ that at each stage takes the current public and secret states, updates its secret sate and outputs a new public state.
A learning algorithm $\L$ for $(\G,\D)$ is given the initial public state $p_0$ ($\L$ does not get the initial secret state $s_0$) and at each round $i$:
i) $\L$ either outputs prediction of the conditional distribution $D_{i+1}^{s_0}(p_0,\ldots,p_i)$ of the public output of $\D$ given the initial secret $s_0$ and the observed public states $p_0,\ldots,p_i$, or
ii) $\L$ proceeds to round $i+1$ after observing a new public state $p_{i+1}\gets D_{i+1}^{s_0}(p_0,\ldots,p_i)$.
The goal of the learning algorithm is to output a hypothesis (in a form of a distribution) that is with high probability close in statistical distance to $D_{i+1}^{s_0}(p_0,\ldots,p_i)$.
In other words, $\L$ is trying to be competitive to somebody who knows the initial secret state $s_0$.
In the setting where $\G,\D$ are efficiently constructible Naor and Rothblum \cite{DBLP:conf/icml/NaorR06} gave an algorithm $\L$ that learns $s_0$ in probabilistic polynomial time provided that one-way functions do not exist.
Moreover, their algorithm outputs a hypothesis after seeing a number of samples proportional to the entropy of the initial secret state.

\begin{theorem}[Naor and Rothblum~\cite{DBLP:conf/icml/NaorR06}]\label{thm:NR06}
Almost one-way functions exist if and only if there exists an adaptively changing distribution $(\G,\D)$ and polynomials $\eps(n),\delta(\epsilon)$ such that it is hard to $(\delta(n),\epsilon(n))$-learn the adaptively changing distribution $(\G,\D)$ with $O\left(\delta^{-2}(n)\cdot\eps^{-4}(n)\cdot\log|S_{init}|\right)$ samples.
\end{theorem}

The strategy of the column player (Colin) with limited randomness gives rise to a natural adaptively changing distribution and we show that the algorithm of Naor and Rothblum~\cite{DBLP:conf/icml/NaorR06} can be used to construct a computationally efficient strategy for the row player (Rowena) that achieves utility noticeably larger than the value of the stage game.
Hence, if one-way functions do not exist, then in repeated two-player strategic games there are no computational Nash equilibria with strategies that use sub-linear randomness in the number of the stages.

\begin{theorem}\thmlab{EfficientAdvantage}
Let $G=\langle (A_1,A_2), u\rangle$ be a two-player zero-sum strategic game with no weakly dominant pure strategies and with value $v$. If almost one-way functions do not exist then for any strategy $\{C_n\}_{n\in\NN}$ of Colin in the repeated game of $G$ that uses $o(n)$ random bits, there exists a polynomial time strategy of Rowena with expected average utility $v+\delta(n)$ against $\{C_n\}_{n\in\NN}$ for some noticeable function $\delta$.
\end{theorem}
\begin{proof}
Let $\{C_n\}_{n\in\NN}$ be an arbitrary strategy of Colin that takes $s(n)\in o(n)$ random bits.
Let $\mu$ be the minmax strategy of Rowena in $G$. We define the following adaptively changing distribution $(\G,\D)$.
The generating algorithm $\G$ on input $1^n$ outputs a random string of length $s(n)$ as the initial secret state $s_0$ and the initial history $\emptyset$ of the $n$-stage repeated game of $G$ as the initial public state $p_0$.
The sampling algorithm $\D$ outputs the new secret state $s_{i+1}$ identical to the secret state $s_i$ that it received as an input (i.e., the secret state remains fixed as the $s(n)$ random coins $s_0$) and updates the input public state $p_i$ in the following way.
The sampling algorithm parses $p_i$ as a history of length $i$ in the $n$-stage repeated game of $G$ and computes Colin's action $c_i=C_n(p_i,s_i)$ at $p_i$ using randomness $s_i$.
$\D$ additionally samples Rowena's action $r_i\leftarrow\mu$ according to her minmax strategy and then outputs the history $(p_i,(r_i,c_i))$ of length $i+1$ as the new public state $p_{i+1}$.
Note that after sampling the initial secret state $s_0$ the only randomness used by $\D$ is to sample the minmax strategy of Rowena.

It follows from Theorem~\ref{thm:NR06} that there exists an efficient learning algorithm $\L$ that after at most $k=k(n)\in O(s(n)\cdot\delta^{-2}(n)\epsilon^{-4}(n))$ samples from $\D$ outputs a hypothesis $h$ such that
$
\Pr[\SD(D_{k+1}^{s_0},D_{k+1}^{h})\leq \epsilon(n)]\geq 1-\delta(n).
$
Consider the strategy of Rowena that uses $\L$ in order to learn Colin's random coins. In particular, a strategy that at each stage $i$ runs $\L$ on the current history $p_{i-1}$ and if $\L$ outputs some hypothesis $h$ then the strategy plays the best response to Colin's action at stage $i$ sampled according to $D_{i+1}^{h}$; and otherwise it plays according to Rowena's minmax strategy $\mu$.
This strategy can be efficiently implemented and it achieves expectation at least $v$ in the $n-1$ stages in which Rowena plays according to her minmax strategy.\footnote{Note that if $\L$ does not output a hypothesis at the current stage, then Rowena chooses her action according to the same distribution as in $\D$, her minmax strategy, and her expectation is $v$.}
It remains to show that Rowena has a noticeable advantage over the value of the game at the stage in which $\L$ outputs the hypothesis $h$ about $s_0$ and Rowena selects her strategy as the best response to Colin's action sampled according to $D_{k+1}^{h}$.

First, note that since $G$ has no weakly dominant strategies, the best response to any pure action $a_2$ of Colin achieves a positive advantage over the value of the game. This observation follows from the fact that Rowena's minmax strategy achieves expectation at least $v$ against any action of Colin and from the fact that the minmax strategy must be mixed (as there are no weakly dominant strategies). By moving all the probability in the minmax strategy to the action with highest payoff given that Colin plays $a_2$, Rowena achieves a value strictly larger than $v$.
Hence, there exists some constant $e$ (depending only on $G$) such that if $D_{k+1}^{h}$ is $e$-close in statistical distance to $D_{k+1}^{s_0}$ then the expectation of the best response against $D_{k+1}^{h}$ achieves expectation at least $v+c$ for some constant $c>0$. Moreover, it is good enough if $\L$ outputs such $h$ with probability at least $1-\delta$ for some constant $\delta>0$. Since $\epsilon$ and $\delta$ can be constant, for all large enough $n$ the learning algorithm $\L$ outputs the hypothesis after receiving at most $k<n$ samples which allows Rowena to get expectation at least $v+\frac{1}{n}c$.
\end{proof}

It follows from \thmref{EfficientAdvantage} that if one-way functions do not exist, then there is no computational Nash equilibrium of repeated two-player zero-sum games where one of the players uses random strings of length sub-linear in the number of stages.

\begin{corollary}\corlab{NoComputationalNash}
Let $G=\langle (A_1,A_2), u\rangle$ be a two-player zero-sum strategic game with no weakly dominant pure strategies and with value $v$. If almost one-way functions do not exist then there is no computational Nash equilibrium of the repeated game of $G$ in which strategy of one of the players uses $o(n)$ random bits.
\end{corollary}
\begin{proof}Assume that there exists a computational Nash equilibrium $(\{C_{n}^{1}\}_{n\in\NN},\{C_{n}^{2}\}_{n\in\NN})$ of $\{G^n\}_{n\in\NN}$, the repeated game of $G$, in which the strategy of one of the players uses random strings of length $o(n)$. Without loss of generality, let Colin be the player with strategy that uses sub-linear randomness in the number of stages.

Denote by $w(n)$ the expectation of Rowena in this computational Nash equilibrium, i.e., for all $n\in\NN$, $w(n)=\E[u_{n}^{*}(C_{n}^{1},C_{n}^{2})]$.
First, consider the case when $w(n)\leq v+\eta(n)$ for some negligible function $\eta$.
By \thmref{EfficientAdvantage} there exists a polynomial-time strategy of Rowena that achieves expectation $v+\delta(n)$ against $\{C_{n}^{2}\}_{n\in\NN}$ for some noticeable function $\delta$. Thus, this strategy constitutes Rowena's computationally efficient deviation to the above strategy profile that is profitable by some non-negligible amount.
Second, consider the case when $w(n)=v+\delta(n)$ for some noticeable function $\delta$. Colin can efficiently approximate the strategy that at each stage achieves his minmax payoff profile in the stage game to achieve expected payoff in the repeated game at least $-v-\eta(n)$, where $\eta$ is a negligible function. Such strategy constitutes Colin's computationally efficient deviation that achieves non-negligible advantage over the above utility profile. In both cases, $(\{C_{n}^{1}\}_{n\in\NN},\{C_{n}^{2}\}_{n\in\NN})$ is not a computational Nash equilibrium of the repeated game of $G$.
\end{proof}

%---------------- End of Low Entropy Computational Nash -----------------%

%--------------- End of Nash Equilibria with Low Entropy ----------------%

%------------------------ Strong Exploitation ---------------------------%
\section{Strong Exploitation of Low-Entropy Opponents}\seclab{StrongExploitation}
We showed in the previous sections that equilibrium strategies in repeated two-player zero-sum games (both with or without restrictions on the computational power of the players) require entropy at least linear in the number of stages.
A natural approach for enabling equilibria that require lower amount of randomness might be to relax the solution concept and consider $\eps$-Nash equilibria, i.e., to ask what is the amount of randomness necessary for equilibrium strategies when the players are indifferent to improvements in utility smaller than $\eps$.

As can be seen from the following argument, an equivalent question is how much can a player exploit an opponent that uses a strategy of low-entropy.
Let $\alpha$ be an entropy level such that Rowena can exploit any Colin's strategy of entropy below $\alpha$ by more than $\epsilon$ (i.e., she can achieve expected utility in the repeated game improving by at least $\eps$ over the value of the stage game).
Then in any $\eps$-Nash equilibrium of the repeated game the strategy of the column player must be of entropy at least $\alpha$.

%------------------- Exploitation with Unbounded Strategies -------------------%
\subsection{Computationally Unbounded Players}
The performance of strategies with bounded entropy in repeated two-player zero-sum games was previously studied in the standard setting with players that do not face any computational limitations.
Towards this direction, Neyman and Okada~\cite{neyman1999strategic} introduced a notion of \emph{strategic entropy} in the context of repeated two-player zero-sum games in order to analyze repeated games played by bounded automata or players with bounded recall.
Subsequently, \cite{DBLP:journals/geb/NeymanO00} gave an asymptotic characterization of the value of repeated two-player zero-sum games when one of the players is restricted to strategies of bounded strategic entropy.
In particular, they showed that if the row player can use strategies of strategic entropy at most $\gamma n$, then in the $n$-stage game she can guarantee expected average utility at most $(\cav U)(\gamma)$; where $U(\gamma)$ is the maximal expected utility the row player can guarantee in the stage game by a strategy of entropy at most $\gamma$, and $\cav U$ is the concavification of $U$ (i.e., the smallest concave function larger or equal to $U$ for all $\gamma\geq 0$).

\paragraph{Repeated matching pennies.}For the special case of the repeated game of matching pennies (given in Figure~\ref{fig:MP}), Budinich and Fortnow~\cite{DBLP:conf/sigecom/BudinichF11} noticed a smooth tradeoff between the amount of entropy available to players and the necessary relaxation of the Nash equilibrium solution concept.
In particular, they showed that in any $\eps$-Nash equilibrium of the $n$-stage repeated game of matching pennies the players must use strategies of entropy at least $(1-\eps)n$ (for all $0\leq\eps\leq 1$).
Their result follows by observing that in the $n$-stage game of matching pennies for all $0\leq\eps\leq 1$, the best response of the column player to any strategy of the row player of entropy at most $(1-\epsilon)n$ achieves expected utility at least $\epsilon$.
This observation can be derived from the result of Neyman and Okada~\cite{DBLP:journals/geb/NeymanO00} by noticing that in the one-shot game of matching pennies $(\cav U)(1-\eps)=-\eps$.
Hence, in the $n$-stage game of matching pennies the row player can guarantee for herself average expected utility at most $(\cav U)(1-\eps)=-\eps$ by a strategy of entropy at most $(1-\eps)n$, and equivalently the column player can achieve expectation at least $\eps$.

In fact, the result of Neyman and Okada~\cite{DBLP:journals/geb/NeymanO00} implies that the relation between $\eps$-Nash equilibria and the entropy of the players' strategies can be extended to all repeated two-player zero-sum games.
\begin{theorem}\thmlab{GeneralAdvantage}
Let $G=\langle (A_1,A_2), u\rangle$ be a two-player zero-sum strategic game of value $v$ and let $\beta>0$ denote the minimal entropy of a minmax strategy for the column player in $G$. For any $0<\eps\leq 1$, there exists $c>0$ such that if $\sigma$ is a strategy of the column player of entropy $(1-\eps)\beta n$ in the $n$-stage repeated game of $G$ then the row player has a deterministic strategy that achieves average payoff of at least $v+c$ against $\sigma$.
\end{theorem}
For completeness we provide the proof of \thmref{GeneralAdvantage} in Appendix~\ref{sec:GeneralAdvantage}.

\paragraph{Limits on exploiting a low-entropy opponent in non-zero-sum games.}
In repeated non-zero-sum games, unlike in repeated two-player zero-sum games, it is in general not possible for a player to always achieve utility strictly above his minmax level given that his opponent uses low-entropy strategy.
We illustrate this phenomenon on the game $G$ given by the payoff matrix in Figure~\ref{fig:ExtendedMatchingPennies} that we discussed in Section~\ref{sec:unbounded}.
Note that if Colin plays his pure action ``left'' then Rowena gets utility $0$, her minmax payoff, irrespective of her strategy. Even though Colin needs at least one random bit to play his equilibrium strategy in $G$, Rowena cannot benefit from the imperfect play of her opponent at all.
Note that this limitation occurs even if any strategy of Colin in a Nash equilibrium of the repeated game of $G$ must use randomness linear in the number of stages.

%--------------- End of Exploitation with Unbounded Strategies ----------------%

%------------------- Exploitation with Efficient Strategies -------------------%
\subsection{Computationally Efficient Players}

Our results from \secref{ComputationalGeneral} (i.e., \thmref{EfficientAdvantage}) show that if one-way functions do not exist, then it is possible to efficiently gain a noticeable advantage over an opponent that uses randomness sub-linear in the number of the stages.
We find it as an intriguing open problem to show a stronger version of \thmref{EfficientAdvantage} analogous to know results in the setting with computationally unbounded players (i.e., \thmref{GeneralAdvantage}). In particular, to show that it is possible to efficiently gain \emph{a constant advantage} over an opponent that uses randomness sub-linear in the number of the stages (even for the special case of the repeated game of matching pennies).
%--------------- End of Exploitation with Efficient Strategies ----------------%

%------------------------ End of Strong Exploitation --------------------------%

%------------------------------ BIBLIOGRAPHY ----------------------------------%
%\newpage
%\printbibliography

%------------------------------ APPENDIX --------------------------------------%
%\newpage
\appendix
%------------------------------ Matching Pennies ------------------------------%
\section{Exploiting Low Entropy in Two-Player Zero-Sum Games}\label{sec:GeneralAdvantage}
In this appendix we provide the proof of \thmref{GeneralAdvantage} that establishes that if one player uses a constant fraction less randomness in the repeated two-player zero-sum game, then the other player can obtain an average payoff that is larger than the value of the stage game by a constant.

We use the following lemma about performance of low-entropy strategies in two-player zero-sum games in the proof of \thmref{GeneralAdvantage}.
\begin{lemma}\lemmlab{LowEntropyOneShot}
Let $G=\langle (A_1,A_2), u\rangle$ be a two-player zero-sum strategic game of value $v$ and let $\beta>0$ denote the minimal entropy of a minmax strategy for the column player in $G$.
For every $\eps>0$, there exists $c_\eps>0$ such that if $\sigma$ is a strategy of the column player of entropy $(1-\eps)\beta$ then the row player has a strategy that achieves utility at least $v+c_\eps$ against $\sigma$.
\end{lemma}
\begin{proof}
Let $\sigma$ be an arbitrary strategy of Colin in $G$ of entropy $(1-\epsilon)\cdot\beta$ for some $\eps>0$,
and let $\rho_\sigma$ denote the best response strategy of Rowena to $\sigma$.
First, we show that Rowena's expected utility $\E[u(\rho_{\sigma},\sigma)]$ is at least $v+c$ for some $c>0$.
Suppose to the contrary that Rowena's best response to $\sigma$ achieves expectation at most $v$.
Let $\hat{\rho}$ be the minmax strategy of Rowena in $G$, the profile $(\hat{\rho},\sigma)$ is a Nash equilibrium of $G$:
Rowena's minmax strategy guarantees at least the value of the game $v$. On the other hand, by the hypothesis her best response to $\sigma$ achieves at most $v$, so Rowena's expectation in $(\hat{\rho},\sigma)$ is equal to $v$.
There are no profitable deviations for Colin, since he cannot decrease Rowena's expectation below $v$ given that she plays according to her minmax strategy.
The strategy $\sigma$ of Colin is of entropy $(1-\eps)\cdot\beta<\beta$, and the strategy profile $(\hat{\rho},\sigma)$ is a Nash equilibrium of $G$ contradicting that $\beta$ is the minimal entropy of Colin's strategy in any Nash equilibrium of $G$. Hence, the best response to $\sigma$ must increase Rowena's expectation by a non-zero amount over $v$. The statement of the lemma follows by setting $c_\eps$ to be the infimum of the set of all $c$ achieved against Colin's strategies of entropy $(1-\eps)\cdot\beta$.
\end{proof}

\paragraph{\thmref{GeneralAdvantage}.}
\emph{Let $G=\langle (A_1,A_2), u\rangle$ be a two-player zero-sum strategic game of value $v$ and let $\beta>0$ denote the minimal entropy of a minmax strategy for the column player in $G$. For any $0<\eps\leq 1$, there exists $c>0$ such that if $\sigma$ is a strategy of the column player of entropy $(1-\eps)\beta n$ in the $n$-stage repeated game of $G$ then the row player has a deterministic strategy that achieves average payoff of at least $v+c$ against $\sigma$.}
\begin{proof}
Let $\sigma$ be an arbitrary strategy of the column player (Colin) of Shannon entropy $n\cdot\beta(1-\eps)$ for some $\eps\in [0,1]$.
Let $\rho_\sigma$ be the strategy of the row player (Rowena) that at each non-terminal history $a$ plays the best response in $G$ to Colin's strategy $\sigma(a)$.
Rowena's expectation $\E_{a\gets(\rho_\sigma,\sigma)}[u^{*}(a)]$ is
\[
%\E_{a\gets(\rho_\sigma,\sigma)}[u^{*}(a)]
%=
\frac{1}{n}\left(
\E_{a\gets(\rho_\sigma,\sigma)}[u(a^1)]+ \E_{a\gets(\rho_\sigma,\sigma)}[u(a^2)|a^1] + \cdots
	+ \E_{a\gets(\rho_\sigma,\sigma)}[u(a^n)|(a^1,\ldots , a^{n-1})]\right)\ .
\]
By the definition of conditional expectation, we rewrite her expectation as a summation over all terminal histories, i.e.,
\begin{multline}\label{eqn:sums}
%\E_{a\gets(\rho_\sigma,\sigma)}[u^{*}(a)]
% =
\frac{1}{n}\Bigg(\sum_{b\in A^{n}}
(\rho_\sigma,\sigma)(b)\cdot\bigg(\E_{a\gets(\rho_\sigma,\sigma)}[u(a^1)]+ \E_{a\gets(\rho_\sigma,\sigma)}[u(a^2)|a^1=b^1] + \cdots\\
	+ \E_{a\gets(\rho_\sigma,\sigma)}[u(a^n)|(a^1,\ldots , a^{n-1})=(b^1,\ldots , b^{n-1})]\bigg)\Bigg)\ .
\end{multline}

Note that for every terminal history $b\in A^{n}$ the summands correspond to the expectation of Rowena at the non-terminal subhistories of $b$.
For any terminal history $b\in A^n$, the total sum of entropy used in $\sigma$ at the subhistories of $b$ is at most $(1 - \eps)\beta n$,
which implies that there are at least $n'= n \left(1- (1-\eps)/(1-\frac{\eps}{2})\right)$ subhistories of $b$ where the Colin's strategy has entropy at most $(1 - \frac{\eps}{2})\beta$.
To see this assume that there exists a terminal history $b$ with less than $n'$ subhistories where $\sigma$ uses entropy at most $(1 - \frac{\eps}{2})\beta$.
Then the total entropy of $\sigma$ on all subhistories of $b$ is strictly larger than
\[
(n-n')\left(1 - \frac{\eps}{2}\right)\beta
=
\left(n-n \left(1- \frac{(1-\eps)}{(1-\frac{\eps}{2})}\right)\right)\left(1 - \frac{\eps}{2}\right)\beta
=
(1 - \eps)\beta n\ ,
\]
a contradiction.
As shown in \lemmref{LowEntropyOneShot}, for each subhistory of $b$ where Colin uses strategy of entropy at most $(1-\frac{\eps}{2})\beta$, Rowena's best response achieves at least $v + c$, where $c=c_{\eps/2}>0$ is a value determined by the game $G$ (and a function of epsilon).
On all other subhistories of $b$ (with Colin's strategy of entropy larger than $(1- \frac{\eps}{2})\beta$) the value of Rowena is at least $v$.
Therefore the total utility (the sum of the expectations over all subhistories of $b$) is at least  $n v + c\cdot n' = n(v+c')$, where $c'=c\left(1- (1-\eps)/(1-\frac{\eps}{2})\right)>0$.

Since this holds for every terminal history of $G^n$, it follows from (\ref{eqn:sums}) that the strategy $\rho_\sigma$ of Rowena achieves average expected utility at least $v+c'$ against $\sigma$ in $G^n$.
\end{proof}

Note that the constant $c$ by which the row player can exploit strategy of the column player of entropy $(1-\eps)\beta n$ is related to the possible gain of the row player in the stage game, given that the column player plays strategy of entropy $(1-\eps)\beta$.
To make the connection explicit, we use the following notation from Neyman and Okada~\cite{DBLP:journals/geb/NeymanO00}. Let $G=\langle(A_1,A_2),u \rangle$ be the stage game and for $\gamma\geq 0$ define
\[
U(\gamma)=\max_{\substack{\sigma\in\Delta(A_1)\\ H(\sigma)\leq\gamma}}\min_{a_2\in A_2}\E[u(\sigma,a_2)].
\]
Hence, $U(\gamma)$ is the maximal expected utility the row player can guarantee with a strategy of entropy at most $\gamma$;
or equivalently, $-U(\gamma)$ is the minimal expected utility that the column player can achieve by a best response to any strategy of the row player of entropy at most $\gamma$.
Note that $U(0)$ is equal to the row player's minmax level in pure strategies, and for all $\gamma\geq 0$, $U(\gamma)$ is at most the value of the game.
Using this notation the statement of \thmref{GeneralAdvantage} can be restated as:

 \paragraph{\thmref{GeneralAdvantage} (restated).}
\emph{Let $G=\langle (A_1,A_2), u\rangle$ be a two-player zero-sum strategic game of value $v$ and let $\beta>0$ denote the minimal entropy of a minmax strategy for the row player in $G$.
For any $0<\eps\leq 1$, if $\sigma$ is a strategy of the row player of entropy $(1-\eps)\beta n$ in the $n$-stage repeated game of $G$ then the column player has a deterministic strategy that achieves average payoff of at least $-v-\left(1- \frac{(1-\eps)}{(1-\frac{\eps}{2})}\right)U((1-\frac{\eps}{2})\beta)$ against $\sigma$.}

We remark that an improved bound on the expectation can be obtained using the technique of Neyman and Okada and the column player can in fact achieve average expected utility at least $-v-(\cav U)((1-\eps)\beta)$, where $\cav U$ is the smallest concave function larger or equal than $U$.

\thmref{EpsilonNashZeroSum} below can be seen as  a ``converse'' of \thmref{GeneralAdvantage}. Specifically, we show that even if the players are restricted to strategies of entropy $(1-\eps)\beta n$ then there exists an $\eps'$-Nash equilibrium of $G^n$ for some $\eps'$ proportional to $\eps$.

\begin{theorem}\label{thm:EpsilonNashZeroSum}
Let $G$ be a two-player zero-sum strategic game such that the minimal entropy of a minmax strategy is $\beta>0$ for both players.
There exists $c>0$ such that for all $0<\eps\leq 1$ and for all $n$, there exists a $\left(c\cdot\frac{\lceil n\eps\rceil+1}{n}\right)$-Nash equilibrium of the $n$-stage repeated game of $G$ in which the players' strategies are of entropy at most $(1-\eps)\beta n$.
\end{theorem}
\begin{proof}
Let $\sigma$ be the strategy profile in the $n$-stage repeated game of $G$ in which the players play in the first $\lfloor n(1-\eps)\rfloor$ stages according to their minmax strategies of minimal entropy (i.e., entropy $\beta$), and in the remaining $\lceil n\eps\rceil$ stages the players alternate between playing the (pure) action profiles $a^*\in A_1\times A_2$ and $a^\dagger\in A_1\times A_2$, such that $p^*=u(a^*)$ is the maximum payoff of Rowena in $G$ and $p^\dagger=u(a^\dagger)$ is the minimal payoff of Rowena in $G$. Note that by construction of $\sigma$, the players use strategies of entropy at most $n\cdot\beta(1-\eps)$.

Assume $\lceil n\eps\rceil$ is odd (the argument for $\lceil n\eps\rceil$ even is analogous).
The expected utility of Rowena in $\sigma$ in the $n$-stage repeated game of $G$ is
\[
\E[u^*(\sigma)]
=
\frac{1}{n}\left(\lfloor n(1-\eps)\rfloor\cdot v + \frac{1}{2}(\lceil n\eps\rceil-1)(p^*+p^\dagger) + p^*\right)\ ,
\]
where $v$ is the value of $G$.
The expectation of every deviating strategy $\sigma'_2$ of Colin is
\[
-\E[u^*(\sigma_1,\sigma'_2)]
\leq
\frac{1}{n}\left(-\lfloor n(1-\eps)\rfloor\cdot v + \frac{1}{2}(\lceil n\eps\rceil-1)(-p^\dagger-p^\dagger) - p^\dagger\right)\ ,
\]
hence Colin can increase his utility by at most $\frac{1}{2n}(p^*-p^\dagger)(\lceil n\eps\rceil+1)$.
Similarly, the increase in expectation from any deviating strategy of Rowena can be upper bounded by $\frac{1}{2n}(p^*-p^\dagger)(\lceil n\eps\rceil-1)$.
Therefore, $\sigma$ is a $\left(c\cdot\frac{\lceil n\eps\rceil+1}{n}\right)$-Nash equilibrium of the $n$-stage repeated game of $G$ for $c=\frac{1}{2}(p^*-p^\dagger)$, and the statement of the proposition follows since $\frac{1}{2}(p^*-p^\dagger)$ is a constant independent of $\eps$ and $n$.
\end{proof}

\section{Matching Pennies}\applab{MatchingPennies}
The game of \emph{matching pennies} is a two-player zero-sum strategic game given by the payoff matrix in Figure~\ref{fig:MP}. Both players can either play Heads ($H$) or Tails ($T$). The only Nash equilibrium is the strategy profile $(\frac{1}{2}H+\frac{1}{2}T,\frac{1}{2}H+\frac{1}{2}T)$ in which both players randomize uniformly over $H$ and $T$.
\begin{figure}[h!]%
\vspace{-1em}
\def\arraystretch{1.5}
\begin{center}
\begin{tabular}{r|c|c|}
\multicolumn{1}{r}{}
 &  \multicolumn{1}{c}{Heads ($H$)}
 & \multicolumn{1}{c}{Tails ($T$)} \\
\cline{2-3}
~~~Heads ($H$)~~~ & $~~1,-1$ & $-1,~~1$ \\
\cline{2-3}
~~~Tails ($T$)~~~ & $-1,~~1$ & $~~1,-1$ \\
\cline{2-3}
\end{tabular}
\end{center}
\vspace{-1em}
\caption{The payoff matrix of the game of matching pennies.}%
\label{fig:MP}%
\vspace{-1em}
\end{figure}
By \thmref{LinearEntropy}, in the equilibrium for the $n$-stage repeated game of matching pennies both players randomize uniformly between playing Heads and Tails at each stage, and the entropy of the equilibrium strategy of each player is exactly $n$.

We now give a generalization of Lemma~3.1 from Budinich and Fortnow~\cite{DBLP:conf/sigecom/BudinichF11}.

\begin{theorem}\thmlab{LowEntropyNoNashMP}
For any $\eps\in[0,1]$, let $\sigma$ be a strategy of the column player of entropy $n(1-\eps)$ in the $n$-stage repeated game of matching pennies. The row player has a deterministic strategy that achieves payoff of at least $\eps$ against $\sigma$.
\end{theorem}
\begin{proof}
Let $\rho_\sigma$ be the strategy of the row player (Rowena) that finds the most likely action of the column player (Colin) at each history and plays the best response to that action. For any stage $t=1,\ldots ,n$, and any terminal history $a\in A^n$ we denote by $p_{a}^{t}$ the probability of Colin's most likely action at stage $t$ at the subhistory $(a^1,\ldots,a^{t-1})$.

Consider the following function $\varphi: A^n\times \{1,\ldots ,n+1\}\rightarrow\RR$ defined for any terminal history $a\in A^n$ and any $t\in\{1,\ldots ,n, n+1\}$ as:
\[
\varphi(a,t)=\sum_{i=1}^{t-1}u(a^i) - H(\sigma_{a}^{t})\ ,
\]
where $\sigma_{a}^{t}\in\Delta(\times_{j=t}^{n}A_2)$ is the distribution of the actions taken by Colin in $\sigma$ at stages $t,\ldots , n$ given the history of the play up to stage $t$ is $(a^{1},\ldots,a^{t-1})$.  Note that for $t = n+1$, Colin has no more actions to take, and by convention we write $H(\sigma_{a}^{n+1}) = 0$, so that $\varphi(a,n+1)=\sum_{i=1}^{n}u(a^{i})$ (i.e.~the total accumulated utility of Rowena at the terminal history $(a^{1}, \dots, a^{n})$.  Also note that for any terminal history $a$ the value of $\varphi(a,1)$ is $-H(\sigma_{a}^{1})=-H(\sigma)$, i.e., minus entropy of the distribution $\sigma$ of Colin's play in all the $n$ stages.

Now consider the expected increase in $\varphi$ between two consecutive stages when Colin's actions are drawn from $\sigma$ and Rowena's actions are chosen according to $\rho$, i.e., for every $t\in\{1,\ldots,n\}$ consider
\[
\E_{a\gets(\rho_\sigma,\sigma)}\left[\varphi(a,t+1)-\varphi(a,t)\right]\ .
\]
We expand the above using the definition of $\varphi$ and get
\[
\E_{a\gets(\rho_\sigma,\sigma)}\left[\left(\sum_{i=1}^{t}u(a^i)-\sum_{i=1}^{t-1}u(a^i)\right)+\left(-H(\sigma_{a}^{t+1})+ H(\sigma_{a}^{t})\right)\right]\ .
\]
Which can be simplified using the probability of the most likely action of Colin at history $(a_1,\ldots, a_{t-1})$ as
\[
\E_{a\gets(\rho_\sigma,\sigma)}\left[2p_{a}^{t}-1+\left(-H(\sigma_{a}^{t+1})+ H(\sigma_{a}^{t})\right)\right]\ .
\]
We can expand the first entropy term
\[
\E_{a\gets(\rho_\sigma,\sigma)}\left[2p_{a}^{t}-1+\left(-\left(p_{a}^{t}\cdot{}H(\sigma_{a}^{t}|a_{2}^{t}=\heartsuit))+(1-p_{a}^{t})\cdot{}H(\sigma_{a}^{t}|a_{2}^{t}=\spadesuit)\right)+ H(\sigma_{a}^{t})\right)\right]\ ,
\]
where $\heartsuit$ denotes the most likely action of Colin at stage $t$ after history $(a^1,\ldots,a^{t-1})$ and $\spadesuit$ denotes its alternative.
We can rewrite the expression using the definition of conditional entropy to
\[
\E_{a\gets(\rho_\sigma,\sigma)}\left[2p_{a}^{t}-1+\left(-H(\sigma_{a}^{t}|\varsigma_{a}^{t})+ H(\sigma_{a}^{t})\right)\right]\ ,
\]
where $\varsigma_{a}^{t}\in\Delta(A_2)$ denotes the distribution of Colin's action at stage $t$ after the history $(a^1,\ldots,a^{t-1})$.
Because of the chain rule for conditional entropy we get that
\begin{align*}
\E_{a\gets(\rho_\sigma,\sigma)}\left[\varphi(a,t+1)-\varphi(a,t)\right]
&=
\E_{a\gets(\rho_\sigma,\sigma)}\left[2p_{a}^{t}-1+H(\varsigma_{a}^{t})\right]\\
&\geq
\E_{a\gets(\rho_\sigma,\sigma)}\left[2p_{a}^{t}-1+(-2p_{a}^{t}+2))\right]\\
&\geq 1\ .
\end{align*}

Finally, we use the above lower bound on the expected increase of $\varphi$ to bound the expectation of Rowena when the players play according to the strategy profile $(\rho_\sigma,\sigma)$
\begin{align*}
\E_{a\gets(\rho_\sigma,\sigma)}\left[u^{*}(a)\right]\cdot n
&=
\E_{a\gets(\rho,\sigma)}\left[\varphi(a,n+1)\right]\\
&\geq
\E_{a\gets(\rho_\sigma,\sigma)}\left[\varphi(a,1)\right] + n\cdot\min_{t\in[n]}\left\{\E_{a\gets(\rho_\sigma,\sigma)}\left[\varphi(a,t+1)-\varphi(a,t)\right]\right\}\\
&\geq
 -H(\sigma) + n=-n(1-\eps)+n=n\eps\ .
\end{align*}
Therefore, the expected average payoff of Rowena is at least $\eps$.
\end{proof}

We give also an alternative and more straightforward proof of \thmref{LowEntropyNoNashMP} that follows the structure of the proof of \thmref{GeneralAdvantage}.
\begin{proof}[Proof of \thmref{LowEntropyNoNashMP} (alternative)]
Let $\sigma$ be an arbitrary strategy of Colin of Shannon entropy $n(1-\eps)$ for some $\eps\in [0,1]$.
Let $\rho_\sigma$ be the strategy of Rowena that at each non-terminal history $a$ plays the best response to Colin's strategy $\sigma(a)$. We can express Rowena's expectation $\E_{a\gets(\rho_\sigma,\sigma)}[u^{*}(a)]$ as
\[
%\E_{a\gets(\rho_\sigma,\sigma)}[u^{*}(a)]
%=
\frac{1}{n}\left(\E_{a\gets(\rho_\sigma,\sigma)}[u(a^1)]+ \E_{a\gets(\rho_\sigma,\sigma)}[u(a^2)|a^1] + \cdots
	+ \E_{a\gets(\rho_\sigma,\sigma)}[u(a^n)|(a^1,\ldots , a^{n-1})]\right)\ ,
\]
which can be rewritten due to the definition of conditional expectation as a summation over terminal histories
\begin{multline*}
%\E_{a\gets(\rho_\sigma,\sigma)}[u^{*}(a)]
% =
\frac{1}{n}\Bigg(\sum_{b\in A^{n}}
(\rho_\sigma,\sigma)(b)\cdot\bigg(\E_{a\gets(\rho_\sigma,\sigma)}[u(a^1)]+ \E_{a\gets(\rho_\sigma,\sigma)}[u(a^2)|a^1=b^1] + \cdots\\
	+ \E_{a\gets(\rho_\sigma,\sigma)}[u(a^n)|(a^1,\ldots , a^{n-1})=(b^1,\ldots , b^{n-1})]\bigg)\Bigg)\ .
\end{multline*}

For every terminal history $b=(b_1,\ldots,b_{n})$, the total entropy of $\sigma$ over the non-terminal subhistories of $b$ is bounded by $n(1-\eps)$, i.e.,
\begin{equation}\label{eqn:entropySummands}
H(\sigma(\emptyset))+\sum_{i=1}^{n-1}{H(\sigma(b_1,\ldots,b_i))}\leq n(1-\eps)\ .
\end{equation}
We define $\eps_0=1-H(\sigma(\emptyset))$ and for every $i\in\{1,\ldots,n-1\}$ we define $\eps_i=(1-H(\sigma(b_1,\ldots,b_i)))$.
Note that $0\leq\eps_i\leq 1$ for every $i\in\{0,\ldots,n-1\}$ and from inequality~(\ref{eqn:entropySummands}) we get that $\eps\leq\frac{1}{n}\sum_{i=0}^{n-1}\eps_i$.
In order to conclude that Rowena's expected utility in the strategy profile $(\rho_\sigma,\sigma)$ is at least $\epsilon$, it is sufficient to show that for every subhistory $b'$ of $b$ the expectation $\E[u(\rho_{\sigma}(b'),\sigma(b'))]$ is least $1-H(\sigma(b'))$.

For an arbitrary non-terminal history $h$, consider Rowena's expectation in $G$ given the strategy profile $(\rho_\sigma(h),\sigma(h))$. Since $\rho_\sigma(h)$ is the best response to $\sigma(h)$, Rowena's expectation is $2p-1$, where $p$ is the probability of Colin's most probable action at history $h$. We need to show that for all $p\in[1/2,1]$
\[
2p-1 \geq 1 - H(\sigma(h))=1+p\log_2(p)+(1-p)\log_2(1-p)\ .
\]
For $p$ equal $1/2$ or $1$, the left side and the right side of the inequality are equal. Since $2p-1$ is a linear function and $1+p\log_2(p)+(1-p)\log_2(1-p)$ is a convex function on $[1/2,1]$, the inequality holds. This concludes the proof.
\end{proof}

It follows form \thmref{LowEntropyNoNashMP} that if the players can use only strategies of entropy $(1-\epsilon)n$ (i.e., lower than $n$-times the entropy of an equilibrium of the single-shot matching pennies) then Nash equilibria in the $n$-stage repeated game of matching pennies do not exist.

\begin{proposition}\label{thm:MPtradeoffs}
Let $G^n$ be the $n$-stage repeated game of matching pennies.
\begin{enumerate}
\item
 For all $0\leq\eps\leq 1$, if $\sigma$ is an $\eps$-Nash equilibrium of $G^n$ then the players' strategies in $\sigma$ are of entropy at least $n(1-\eps)$.
\item
 For all $0\leq\eps\leq 1$, there exists an $(\eps+\frac{2}{n})$-Nash equilibrium of $G^n$ in which the players' strategies are of entropy at most $(1-\eps)n$.
\end{enumerate}
\end{proposition}
\begin{proof}
First, we show that any $\eps$-Nash equilibrium $\sigma$ in the $n$-stage repeated game of matching pennies comprises of strategies of entropy at least $(1-\eps)n$. Assume that there is an $\eps$-Nash equilibrium in which both players use a strategy of strictly smaller entropy than $(1-\eps)n$, i.e., of entropy $(1-\eps')n$ for some $\eps'>\eps$. By \thmref{LowEntropyNoNashMP}, each player $i$ has a strategy $\sigma'_i$ that achieves at least $\eps'$ against $\sigma_{-i}$. Since $\sigma$ is an $\eps$-Nash equilibrium then for any player $i$
\[
	\E[u^{*}_{i}(\sigma)]\geq\E[u^{*}_{i}(\sigma'_i,\sigma_{-i})]-\eps \geq \eps' - \eps > 0\ .
\]
This implies that for both players $\E[u^{*}_{i}(\sigma)] >0$, however it cannot be the case that the expectation of both players is strictly larger than zero, since matching pennies is a zero-sum game.

Second, we show that if the players can use strategies of entropy $(1-\eps)n$ then there exists an $(\eps+\frac{2}{n})$-Nash equilibrium of the $n$-stage repeated game of matching pennies.
To see this, consider a strategy profile in which the players play uniformly at random $H$ and $T$ in the first $\lfloor (1-\eps)n\rfloor$ stages and in the remaining $\lceil\eps n\rceil$ stages Rowena plays always $H$ and Colin alternates between $T$ and $H$ (i.e., the outcome at stage $\lfloor (1-\eps)n\rfloor + 1$ is $(H,T)$).
If $\lceil\eps n\rceil$ is odd then Rowena's expectation is $-\frac{1}{n}$ and otherwise it is $0$.
Both Colin and Rowena can improve their expectation only in the last $\lceil\eps n\rceil$ stages by matching/countering the opponent, but any such deviation can achieve utility at most
\[
\frac{\lceil\eps n\rceil}{n} \leq \frac{\eps n + 1}{n}\leq\eps + \frac{1}{n}\ .
\]
Hence, both players can improve the utility by at most $\eps + \frac{2}{n}$ by deviating from the prescribed strategy profile, and it constitutes an $(\eps + \frac{2}{n})$-Nash equilibrium.
\end{proof}

\subsection{Matching Pennies with Computationally Efficient Players}
In this section we prove the statement of \thmref{EfficientAdvantage} for the special case of the game of matching pennies
without relying on the framework of adaptively changing distributions of Naor and Rothblum~\cite{DBLP:conf/icml/NaorR06},
but using the classical results on pseudorandomness discussed in \secref{CryptoNotation}.
In particular, that if one-way functions do not exist, then the players cannot efficiently generate unpredictable sequences of bits using only a few truly random bits.
Hence, in the repeated game of matching pennies any player can at some stage efficiently predict and exploit the next move of an opponent that uses amount of random bits sub-linear in the number of stages.

\begin{theorem}\thmlab{EfficientAdvantageMP}
If one-way functions do not exist then for any polynomial-size circuit family $\{C_n\}_{n\in\NN}$ implementing a strategy of Colin in the repeated game of matching pennies using at most $n-1$ random bits, there exists a polynomial time strategy of Rowena with expected utility $\delta(n)$ against $C_n$ for some noticeable function $\delta$.
\end{theorem}
\begin{proof}
Let $\{X_n\}_{n\in\NN}$ be a probability ensemble defined for all $n$ as the random variable over $2n$-bit strings corresponding to the terminal histories in the $n$-stage repeated matching pennies (where $H$ corresponds to $0$ and $T$ to $1$) when Rowena plays uniformly at random and Colin plays according to $C_n$.
Note that $X_n$ is of length $2n$ and it can be generated in polynomial time given at most $2n-1$ random bits, since Colin's strategy uses random strings of length at most $n-1$.

Since one-way functions do not exist, the ensemble $\{X_n\}_{n\in\NN}$ cannot be pseudorandom. In particular, it cannot be unpredictable in polynomial time in the following sense. There exists a polynomial time predictor algorithm $A$ that reads $x\gets X_n$ bit by bit and succeeds in predicting the next value with probability noticeably larger than one half. Formally, let $\next_A(x)$ be a function that returns the $i$-th bit of $x$ if on input $(1^{|x|},x)$ algorithm $A$ reads only the first $i-1 < |x|$ bits of $x$, and returns a uniformly chosen bit in case $A$ reads the entire string $x$. There exists a predictor algorithm $A$ and some positive polynomial $p$, such that
\[
	\Pr[A(1^{|X_n|},X_n) = \next_A(X_n)] \geq \frac{1}{2} + \frac{1}{p(n)}\ ,
\]
where the probability is taken over the randomness of $A$.

We show that Rowena can guarantee for herself at least noticeable expected utility by emulating $A$ on the transcript of the repeated game.
Consider the strategy $R_A$ of Rowena that at each stage $i$ samples a uniformly random bit $r_i$, and if $A(b_i)$ outputs any prediction $c^{*}_i$ of Colin's action then Rowena plays $c^{*}_i$ (to match Colin) and otherwise it plays $r_i$ and uses the action played by Colin at stage $i$ as the next input to $A$. After the stage in which $A$ outputs a prediction $R_A$ plays uniformly at random. The expectation of Rowena can be lower bounded in the following way:
\begin{multline*}
\E[u^*(R_A,C)]
\geq
\frac{1}{n}\Bigg(\Pr[A\text{ outputs } c^{*}_{i}]\cdot \left((n-1)\cdot 0+ 2\left(\frac{1}{2}+\frac{1}{p(n)}\right)-1\right)\\
 + (1-\Pr[A\text{ outputs } c^{*}_{i}])\cdot 0\Bigg)\ .
\end{multline*}
Recall that the actions of Rowena are chosen uniformly at random and the predictor $A$ has to guess a  uniformly random bit if it reads the whole terminal history $x\gets X_n$.
Hence, in order to gain noticeable advantage over one half, $A$ must output its prediction to one of the actions of Colin with at least noticeable probability, i.e., $\Pr[A\text{ outputs } c^{*}_{i}]$ is at least $\delta'(n)$ for some noticeable function $\delta'$.
Thus, the strategy $R_A$ achieves expectation at least $\delta(n)=\delta'(n)\cdot(n\cdot p(n))^{-1}$, which is a noticeable function of $n$.
\end{proof}

%------------------------------ End of Matching Pennies -----------------------%

\end{document}